\documentclass[11pt,a4paper]{article}
\pdfoutput=1
\usepackage[utf8]{inputenc}
\usepackage[hyphens]{url}
\usepackage{graphicx}
\usepackage{natbib} 
\usepackage{caption}

\usepackage[margin=1in]{geometry}
\usepackage{algorithm}
\usepackage{algorithmic}
\usepackage{amsmath}
\usepackage{amssymb}
\usepackage{xspace}
\usepackage{amsthm}
\usepackage{thmtools}
\usepackage{tikz}
\usepackage{soul}
\usepackage{xspace}
\usepackage[svgnames]{xcolor}
\usepackage{hyperref}
\hypersetup{colorlinks={true},urlcolor={blue},linkcolor={DarkBlue},citecolor={black}}
\usepackage{cleveref}
\usepackage{booktabs}
\usepackage{mathtools}
\usepackage{complexity}
\usepackage{cancel}
\usetikzlibrary{calc,shadows.blur,positioning,fit,decorations.pathreplacing}

\author{
  \begin{tabular}{c}
    \textbf{Paul W. Goldberg}\\
  \small{University of Oxford}\\
  \href{mailto:paul.goldberg@cs.ox.ac.uk}
  {\small{\texttt{paul.goldberg@cs.ox.ac.uk}}}
  \end{tabular}
  \and
  \begin{tabular}{c}
    \textbf{Alexandros Hollender}\\
    \small{Technical University of Munich}\\
    \href{mailto:alexandros.hollender@tum.de}
    {\small{\texttt{alexandros.hollender@tum.de}}}
  \end{tabular}
  \and 
  \begin{tabular}{c}
    \textbf{Giannis Tyrovolas}\\
  \small{University of Oxford}\\
  \href{mailto:giannis.tyrovolas@cs.ox.ac.uk}
  {\small{\texttt{giannis.tyrovolas@cs.ox.ac.uk}}}
  \end{tabular}
}

\newtheorem{theorem}{Theorem}

\theoremstyle{definition}
\newtheorem{lemma}[theorem]{Lemma}
\newtheorem{definition}[theorem]{Definition}

\DeclarePairedDelimiter\abs{\lvert}{\rvert}

\newcommand{\PromiseSPNE}[1]{\textsc{PromiseRoundRobinSPNE}\ensuremath{_{#1}}}
\newcommand{\PickOneOf}{\textsc{PickOneOf}}
\newcommand{\OXS}{$\mathit{OXS}$\xspace}
\renewcommand{\phi}{\varphi}
\let\vec\mathbf

\title{Equilibria of Round-Robin: Computational Hardness and Fairness for Few Subadditive Agents}
\date{}
\begin{document}
\maketitle
\thispagestyle{empty}

\begin{abstract}
The round-robin procedure is a simple and well-studied fair division mechanism where agents pick goods in turns. Motivated by draft mechanisms in sports leagues, we investigate strategic behaviour in online round-robin for subadditive agents. This gives rise to an extensive-form game, and we study the computational problem of computing a subgame perfect Nash equilibrium (SPNE). We show that for just two submodular agents, computing an SPNE is $\mathsf{PSPACE}$-hard. Even for the class of $\mathit{OXS}$ utilities, which are a special case of submodular utilities, computing an SPNE remains $\mathsf{NP}$-hard for a small number of agents. We complement our computational results with normative results. We show that for just three additive agents, there exist instances where every equilibrium violates EF1. This separates the online and the direct revelation games. On the positive side, we show that for additive agents every equilibrium allocation is proportional up to one good (PROP1) and for two additive agents it is also EF1. Finally, by showing that round-robin is bossy at equilibrium, we prove that the number of equilibrium allocations can be exponential even if agents have lexicographic preferences.

\end{abstract}

\newpage

\setcounter{page}{1}

\section{Introduction}

Allocating indivisible resources fairly is a fundamental problem in Computational Social Choice and Game Theory.
In fair division, we define notions of fairness and efficiency and construct mechanisms that achieve these notions.
Round-robin is a mechanism that satisfies many fairness guarantees.
In this simple mechanism, agents take turns to pick from available items.
Picking an item makes it unavailable to other agents.

Initially, round-robin was studied in settings where agents have \emph{additive} utilities (e.g., \citet{amanatidis2017approximation, caragiannis2019unreasonable, aziz2022fair}).
Recently, however, there has been increased interest in studying round-robin where agents have \emph{subadditive} valuation functions \citep{barman2020approximation,ghodsi2022fair,amanatidis2023round,tao2024fair,montanari2025weighted,andersen2026computing, connor2026random}.
More than a theoretical challenge, the study of these instances is well-motivated from the drafting of athletes into sports teams.
For example, in the US, the National Basketball Association, National Hockey League and National Football League---often abbreviated as NBA, NHL and NFL---hold an annual draft to allocate new players to teams.
Teams pick from a pool of available players via a generalised version of round-robin.
By a very conservative estimate, the contracts given as part of the NBA draft are worth 100 million USD.\footnote{In 2021 the New Orleans Pelicans sold the 53\textsuperscript{rd} NBA draft pick for 2 million USD \citep{drafttrade}. The value of the 53\textsuperscript{rd} pick is a very conservative lower bound for the value of earlier picks.}

We identify three salient facts about the draft.
Teams have \emph{submodular} preferences: having two players that play the same position has diminishing marginal returns.
Teams are \emph{strategic}: they may use early picks for popular players and later picks for better but less popular players.
That was the tactic of the Seattle Seahawks when drafting Russell Wilson, one of their most successful players.\footnote{\url{https://www.nbcsports.com/nfl/profootballtalk/rumor-mill/news/pete-carroll-praises-john-schneider-for-drafting-russell-wilson}}
The third fact is that the team making the $n$\textsuperscript{th} pick knows which players were picked in the $n-1$ previous picks.\footnote{Teams are given four to five minutes to pick after the previous pick: \url{https://www.nba.com/news/nba-draft-faq}} This makes it an instance of \emph{online} round-robin. This is distinct from the model of \citet{amanatidis2023round, amanatidis2024allocating} where agents directly reveal their preferences to a ``server'' that executes round-robin \emph{offline}.

Motivated by these examples, we study the equilibria of online strategic round-robin with subadditive agents.
As online round-robin is an extensive-form game, the appropriate equilibrium notion is that of a subgame perfect Nash equilibrium (SPNE).
Previous work has studied the complexity of computing SPNEs for additive utilities.
\citet{kohler1971class} present a linear-time algorithm for two agents and \citet{kalinowski2013strategic} show that the problem becomes \PSPACE{}-hard for an unbounded number of agents.
There are no other computational hardness (or tractability) results for a constant number of agents \emph{even for subadditive utilities}. This paper addresses this gap.

\subsection{Our Contributions}

We study the SPNEs of online round-robin for subadditive agents from normative and computational perspectives.

We first investigate the structural complexity of SPNEs. Even for additive utilities and agents holding strict preferences over bundles of items, we show that instances may have exponentially many SPNEs. Additionally, even with the above restrictions, SPNEs violate non-bossiness: an agent is able to affect the bundles of other agents without affecting her own.
We also initiate the normative study of equilibria by constructing an instance with a unique SPNE that violates envy-freeness up to one good (EF1). This is in contrast to the offline setting, where it is known that every pure Nash equilibrium is EF1 \citep{amanatidis2024allocating}. On the positive side, we show that a weaker fairness notion---PROP1---is satisfied in all SPNEs for additive agents.

As our main contribution, we study the computational complexity of computing SPNEs. We show that computing an SPNE is \PSPACE{}-hard even in instances with two agents where one is additive and the other is submodular. This is a substantial increase in computational complexity compared to the linear time algorithm for two additive agents of \citet{kohler1971class}.
If agents have \OXS{} utilities---which are a strict subclass of submodular and a slight generalisation of additive---then for just 11 agents the problem is \NP-hard.
In addition to showing hardness for natural instances with a constant number of agents, we make partial progress in resolving the computational complexity of computing SPNEs for a constant number of additive agents, a \emph{``stubbornly open''} problem in the words of \citet{walsh2020fair}.

\subsection{Related Work}
Round-robin is a simple mechanism for the fair division of indivisible items. 
For a detailed survey of the subject, we refer the reader to the works of \citet{bouveret2016fair} and \citet{amanatidis2023fair}.

There is a substantial literature on the guarantees of round-robin when agents participate truthfully.
For additive utilities, round-robin guarantees EF1 for goods, chores or settings with both goods and chores \citep{caragiannis2019unreasonable,aziz2022fair}.
If agents have disproportionate entitlements in the common resources, as in weighted fair division, then round-robin can be modified to achieve weighted EF1 \citep{chakraborty2021weighted}.
It can also achieve $\frac{1}{2}$-approximate Maximin fair share (MMS) \citep{amanatidis2017approximation}.
For submodular utilities, \citet{barman2020approximation} show that round-robin guarantees $0.21$-approximate MMS which was improved by \citet{ghodsi2022fair} to $\frac{1}{3}$.

Strategic behaviour in round-robin differs between the online and offline settings.
The online setting was first studied for additive utilities by \citet{kohler1971class}.
They provide a linear-time algorithm for computing an SPNE with two additive players.
\citet{kalinowski2013strategic} study a more general setting of sequential allocation where agents pick according to an arbitrary ordering.
Importantly, they also show that for an unbounded number of players---in fact linear in the number of goods---computing an SPNE is \PSPACE-hard.
Although this result is shown for an ordering that is not the alternating ordering of round-robin, it can be easily modified to the round-robin ordering.
Since 2013, there has not been any progress on this front and successive survey papers point out the complexity of computing SPNE for a fixed number of agents as an open problem \citep[Challenge 5]{walsh2016strategic}, \citep{walsh2020fair}.

Offline round-robin is a meaningfully different game.
For additive agents, \citet{aziz2017equilibria} show that a pure Nash equilibrium (PNE) in the form of a ``bluff'' profile always exists and is efficiently computable.
This ``bluff'' profile can be extended to more general utility functions.  \citet{amanatidis2023round} show that the bluff profile is an exact PNE for cancellable utilities and a $\frac{1}{2}$-approximate PNE for submodular utilities. 
However, offline round-robin may not even have a $\frac{3}{4}$-approximate PNE  for submodular functions.
What is surprising in offline round-robin is that the PNEs have fairness guarantees.
For additive agents, any PNE allocation of round-robin is EF1 \citep{amanatidis2024allocating}.
For more general classes of utility functions, weaker but similar guarantees hold at equilibrium \citep{amanatidis2023round}.

\section{Preliminaries}

We study the fair division of indivisible non-negatively valued items. Instances have a set $N$ of agents and a set $M$ of items.
We will write $n = \abs{N}$ and $m = \abs{M}$.
Each agent $A \in N$ has a utility function $u_A \colon 2^{M} \longrightarrow \mathbb{Q}_{\geq 0}$ over bundles of goods.
Utility functions are non-negative and monotone, i.e., $u_A(S) \leq u_A(S')$ whenever $S \subseteq S'$.
Without loss of generality, they are normalised, i.e., $u_A(\varnothing) = 0$.
A special case of utility functions is that of additive utility functions where $u_A(S) = \sum_{j \in S}u_A(\{j\})$.
We will study utility functions that need \emph{not be additive}.

We introduce the following shorthands: for $x \in M$,  we write $u_A(x)$ for $u_A(\{x\})$; for $S \subseteq M$ and $x \in M$, we write $S + x$ for $S \cup \{x\}$ and $S - x$ for $S \setminus \{x\}$.
For an item $x \in M$ and a bundle $S \subseteq M \setminus \{x\}$, we will write $u_A(x \mid S) = u_A(S \cup \{x\}) - u_A(S)$ for the marginal contribution of $x$ to $S$.
For a strategy profile $\boldsymbol{\sigma} = (\sigma_1, \ldots, \sigma_n)$, we write $(\boldsymbol{\sigma}_{-i}, \sigma')$ for the strategy profile $(\sigma_1, \ldots,\sigma_{i-1},\sigma', \sigma_{i+1}, \ldots, \sigma_n)$.

\paragraph{Round-robin.} In round-robin, agents pick items in turns.
Once an agent picks an item, it becomes unavailable to other agents.
Agents pick items based on a fixed ordering, $A_1 > \ldots > A_n$.
If there are more items than agents, agent $A_1$ picks after $A_n$ and the ordering cycles.

Even in the case of two additive agents, agents can benefit from strategising. Consider the example in \Cref{tab:example}. The numbers in the cells represent the value an agent has for a particular item.
Agents pick in order $\mathit{ABA}$ and the agents are additive.
If both players act truthfully and pick their favourite item in each turn, the final allocation will be $S_A = \{x, z\}$ and $S_B = \{y\}$ and agent $A$ receives payoff $4$.
Notice, however, that agent $A$ can strictly improve her utility by picking item $y$ in her first turn.
Then, since agent $B$ has only one remaining turn, it is a dominant strategy for him to pick his favourite available item in the next turn, namely $z$. Finally agent $A$ picks item $x$, resulting in the final allocation $S_A = \{x, y\}$ and $S_B = \{z\}$. So agent $A$'s utility is now 5, a strict improvement of her previous utility.

\begin{table}
    \centering
    \begin{tabular}{lllllll}
    \toprule
        & $x$ & $y$ & $z$ \\ \midrule
    $A$ & 3  & 2   & 1  \\ \midrule
    $B$ & 1   & 3  & 2  \\ \bottomrule
    \end{tabular}
    \caption{An example round-robin instance with agents $A,B$ and items $x, y, z$.}
    \label{tab:example}
\end{table}

Hence, we study round-robin as a game. As this is a sequential game, it is best captured as an extensive-form game.

\begin{definition}
    We denote the possible states of the round-robin game as $Q$.
    Each state is of the form $(i;S; S_1, \ldots, S_n)$ for agent $i \in [n]$, $S, S_1, \ldots, S_n \subseteq M$ . Player $i$ is the player to act, $S$ is the set of available items and $S_i$ are pairwise disjoint partial bundles of each agent at this state.
    At node $(i;S; S_1, \ldots, S_n)$ agent $i$ may take any action $a \in S$, and move to node $({\mathrm{next}(i)}; S -a; S_1,\ldots,S_{i-1},S_i+a, S_{i+1},\ldots, S_n)$ where $\mathrm{next}(i)= i + 1$ if $i < n$ and $\mathrm{next}(i)=1$ if $i = n$.
    The payoff to agent $i$ at a leaf node $(j;\varnothing;S_1, \ldots,S_n)$ is $u_i(S_i)$.
    The initial state of the game, $q_0$ is $(1; M;\varnothing,\ldots, \varnothing)$.
\end{definition}

We write $Q_i$ for the states in which player $i$ acts.
A pure strategy $\sigma_i$ is a mapping from states $Q_i$ to available actions in each state.
Given a strategy profile $\boldsymbol{\sigma}= (\sigma_1, \ldots,\sigma_n)$, the result of the game is determined by starting at $q_0$ and iteratively taking the action given by $\sigma_i$ where $i$ is the owner of the current state until we reach a leaf node.
We write $u_i(\boldsymbol{\sigma})$ as a shorthand for the utility that agent $i$ receives at the leaf node that results from strategy profile $\boldsymbol{\sigma}$.

\begin{definition}[Nash Equilibrium]
    A \emph{pure Nash equilibrium} is a strategy profile $\boldsymbol{\sigma}=(\sigma_1, \ldots, \sigma_n)$ where no agent can benefit by deviating.
    That is, for all $i \in N$, all strategies $\sigma_i'$ of player $i$,
    $u_i(\boldsymbol{\sigma}) \geq u_i(\boldsymbol{\sigma}_{-i},\sigma'_i)$.
\end{definition}

\begin{definition}[SPNE]
    For every node $q$ of an extensive-form game $G$, we define a \emph{subgame} $G_q$.
    $G_q$ is an extensive-form game with initial node $q$ and a strategy space restricted to descendants of $q$ in the game tree.
    A \emph{subgame perfect Nash equilibrium} (SPNE) of a game $G$ is a strategy profile that is a pure Nash equilibrium for every subgame of $G$.
\end{definition}

All finite extensive-form games of perfect information---and in particular round-robin---have an SPNE. 
Notice that the instance given by \Cref{tab:example} also separates the concepts of Nash and SPNE. Consider the strategy profile where players pick the alphabetically first available item.
Then, the final allocation is $S_A = \{x, z\}$ and $S_B = \{y\}$.
This is a Nash equilibrium: if player $A$ does not pick $x$ in his first turn then player $B$ will.
However, the threat of player $B$ is not credible: in the subgame where only $x$ and $z$ are available and it is player $B$'s turn, player $B$ will always pick $z$.
The only SPNE in this game is the one we describe earlier, namely for agents to pick items in order $y \rightarrow z \rightarrow x$.

\paragraph{Utility hierarchy.}
We will repeatedly discuss preferences which are a special case of additive preferences: lexicographic preferences.
\begin{definition}
    Agent $A$ has strict \emph{lexicographic preferences} over items $a_1, \ldots, a_m$ if $A$'s preferences are ordinally equivalent to additive preferences with $u_A(a_k) = 2^{m-k}$. 
\end{definition}
The primary focus of the paper are utilities that are not additive. We will use the hierarchy of complement-free utility functions as given in \citet{lehmann2001combinatorial}.

The smallest strict superset of additive utilities in the hierarchy are \OXS utilities.
\OXS utilities were first introduced by \citet{shapley1962complements} as assignment valuations, and we adapt his presentation to the setting of sports teams.
Suppose a team has $m$ different positions in its formation. 
Let $[n]$ be the set of available athletes to the team. Athlete $i$ provides utility $w_{ij}$ by playing position $j$.
The utility of a team across players playing different positions is additive.
Given a subset of athletes $S \subseteq [n]$, the coach matches athletes to positions so as to maximise the team's utility.
A function $f \colon 2^{M} \rightarrow \mathbb{Q}_{\geq 0}$ is \OXS if it admits this representation.
Hence, \OXS valuations have concise representations.
More formally:

\begin{definition}[\OXS]
    \label{def:OXS}
    A function $f \colon 2^{M} \rightarrow \mathbb{Q}_{\geq 0}$ is \emph{unit demand} if for all $S \subseteq M$, $f(S) = \max_{s \in S}{f(s)}$.
    A function $g \colon 2^{M} \rightarrow \mathbb{Q}_{\geq 0}$ is \OXS if there exist unit demand functions $f_1, \ldots, f_k$, such that for all $S \subseteq M$, $g(S) = \max(\{\sum_{i = 1}^{k} f_i(S_i) \mid S_1, \ldots, S_k \text{ a partition of } S\})$.
\end{definition}

A strictly larger class than \OXS is that of submodular functions.
Submodular functions capture the essential economic property of diminishing marginal returns.

\begin{definition}
    A function $f \colon 2^{M} \rightarrow \mathbb{Q}_{\geq 0}$ is \emph{submodular} if for all $x \in M$ and $S \subseteq T \subseteq M - x$, $f(x \mid S) \geq f(x \mid T)$.
\end{definition}

\section{Structural Properties of SPNEs}

As round-robin is an extensive-form game of perfect information, SPNEs are guaranteed to exist.
However, in the general case, there is no clear way to compute SPNEs other than exhaustive search. We present lemmas that allow us to reason over particular instances.

\begin{restatable}{lemma}{lemmaGreedyAgents}
    Consider state $q$ of round-robin where it is $I$'s turn to play, he is allocated items $S_I'$ and the set of unallocated items is $M'$.
    If it is $I$'s last turn to play, then in any SPNE he will play greedily and pick an item $g\in M'$ that maximises $u_I(S_I' + g)$.
    \label{lemma:agents-are-greedy}
\end{restatable}

\begin{proof}
    The proof is immediate. This is the final opportunity for the agent to play. Hence, he will play the action that maximises his utility in this turn. That means, picking an item with the highest marginal contribution to his utility.
\end{proof}

\begin{restatable}{lemma}{lexicographicSpne}
    \label{lemma:lexicographic-spne}
    Consider a state $q$ of round-robin where it is $I$'s turn to pick an item, he has already been allocated items $S_I'$ and the set of unallocated items is $M'$.
    Suppose there exists $g \in M'$ such that $u_I(S_I' + g) > u_I(S_I' \cup M' - g)$.
    Then, in every SPNE, agent $I$ receives item $g$ in subgame $G_q$.
\end{restatable}
\begin{proof}
    Suppose not. Then, by monotonicity, agent $I$'s utility can be upper bounded by $u_I(S_I' \cup M' - g)$.
    But, agent $I$ can beneficially deviate by picking item $g$ at state $q$ and receive utility which is at least $u_I(S_I' + g) >u_I(S_I' \cup M' - g)$.
\end{proof}

\begin{restatable}{lemma}{lemmaLexicographicRR}
    \label{lemma:lexicographic-rr}
    Consider a state $q$ of round-robin, where the set of available items is $M'$, where $I$,$J$ are players and it is the turn of player $I$ to play.
    Agents $I$ and $J$ are currently allocated items $S_I', S_J'$ respectively and there is an item $g \in M'$ such that $u_I(S_I' + g) > u_I(S_I' \cup M' -g)$ and $u_J(S_J' + g) > u_J(S_J' \cup M' - g)$. Then, in any SPNE, player $I$ must pick item $g$ in state $q$.
\end{restatable}

\begin{proof}
    By \Cref{lemma:lexicographic-spne}, agent $I$ must receive item $g$ in any SPNE allocation. If he does not pick the item in state $q$, then the item will either be picked by another agent or be available at the start of agent $J$'s turn.
    So, \Cref{lemma:lexicographic-spne} applies and $g$ must be in $J$'s allocation as well.
    This is a contradiction as only one agent may receive the item.
\end{proof}

Our first contribution is on the structure of equilibria. 
\citet{kalinowski2013strategic} show that there exist instances where the number of equilibria is exponential as a function of the number of items.\footnote{\citet{kalinowski2013strategic} actually undercount the number of SPNE in their construction. 
They claim that each instance with $4m$ items has $2^m$ equilibria, but using exhaustive search, we found that the number of equilibria for $m \leq 5$ is $2, 6, 25, 140, 1656$.}
However, their construction relies on agents being indifferent among exponentially many bundles.
One could reasonably hope that if agents have strict preferences over \emph{bundles} of items, then the equilibrium allocation is unique.
The following lemma shows that this is \emph{not} the case, even if agents have lexicographic preferences.

\begin{lemma}
    There are instances of round-robin with multiple equilibrium allocations, even for three agents and lexicographic preferences.
    \label{lemma:multiplicity-of-equilibria}
\end{lemma}
\begin{proof}

\begin{restatable}{table}{bossyexample}
    \centering
    \begin{tabular}{lllllll}
    \toprule
        & $u$ & $v$ & $w$ & $x$ & $y$ & $z$ \\ \midrule
    $A$ & 32  & 16   & 8   & 4   & 2 & 1  \\ \midrule
    $B$ & 16   & 32  & 4   & 8   & 2 & 1  \\ \midrule
    $C$ & 8   & 4   & 2   & 32  & 16 & 1  \\ \bottomrule
    \end{tabular}
    \caption{In this instance there are two SPNE allocations: $\{u,w\}, \{v, y\}, \{x, z\}$ and $\{u,w\}, \{v, x\}, \{y, z\}$.}
    \label{tab:bossy}
\end{restatable}

    Consider the instance given by \Cref{tab:bossy} with additive agents $A, B, C$.
    As for each agent $I$, the utilities for individual items are distinct powers of two, we have $u_I(M) \neq u_I(M')$ for all $M, M' \subseteq \{u,v,w,x,y,z\}$ with $M \neq M'$.
    
    Let us study the possible equilibria.
    Suppose agent $A$ picks item $u$ in her first turn.
    Then, we can apply \Cref{lemma:lexicographic-rr} to agents $B, A$ and item $v$ to deduce that agent $B$ must pick item $v$ next.
    We can again apply \Cref{lemma:lexicographic-rr} to agents $C, B$ and item $x$ to deduce that agent $C$ must pick item $x$ next.
    As it is now the agents' last turn, agents all pick greedily by \Cref{lemma:agents-are-greedy}, resulting in the allocation  $\{u,w\}, \{v, y\}, \{x, z\}$ and utilities of $(40, 34, 33)$.

    Using computer aided exhaustive search we can verify that the above is indeed an SPNE and find a second SPNE allocation. The second one is given by:  $A$ picks item $w$, $B$ picks $x$, $C$ picks $y$, $A$ picks $u$, $B$ picks $v$ and $C$ picks $z$.
    This results in an SPNE allocation $\{u,w\}, \{v, x\}, \{y, z\}$ with utility profiles $(40, 40, 17)$.
\end{proof}

By concatenating copies of this instance, we can construct instances with exponentially many equilibria. 
\begin{restatable}{theorem}{exponentialAllocations}\label{theorem:exponential-equilibria}
    There exist round-robin instances with $\Omega(2^k)$ equilibrium allocations, where there are 3 additive agents and $6k$ items and agents have lexicographic preferences.
\end{restatable}

\begin{proof}[Proof sketch]
    We construct an instance with $k$ copies of each good in the instance presented in \Cref{tab:bossy}. We write $G^i = \{u^i, v^i, w^i, x^i, y^i, z^i\}$ for the items in each copy.
    The utility of an agent $I$ for an item $g^i$ is $u_I(g^i) = 2^{6(i-1)}u_I(g)$, where $u_I(g)$ is the utility of agent $I$ for good $g$ in \Cref{tab:bossy}.

    We then show that the SPNE allocations that are possible are exactly the product of SPNE allocations in each instance $G^i$.
    As each instance $G^i$ has two SPNE allocations, this results in $2^k$ SPNE allocations.

    One might be tempted to show that at any SPNE agents must pick items in $G^k$ before picking items in $G^{k-1}$.
    Surprisingly, this is not completely true.
    It is possible at an SPNE for agents to pick items $w^k \rightarrow x^k \rightarrow z^k$ and then for agent $A$ to start picking items in $G^{k-1}$. The reason is that the favourite available item from each agent is different. Hence, each agent can delay picking their favourite item.

    A further complication is that it is not immediate that there are no games in which agents cannot benefit by strategically skipping a turn. In the copied game, this could be simulated by intentionally picking an item not in $G^k$ to essentially skip a turn and result in a new equilibrium allocation.

    A careful case analysis can be found in \Cref{sec:exponential-spne-proof}.
\end{proof}

The equilibria of the instance in \Cref{tab:bossy} feature a further interesting pathology.
Although agent $A$ is indifferent between the two equilibrium allocations, agent $A$'s first pick uniquely determines the equilibrium allocation.
This is a violation of \emph{non-bossiness} \citep{satterthwaite1981bossy}.
Notice, that \emph{non-bossiness} is violated even if we only allow agents to play equilibrium strategies.

\begin{definition}
    Let $\Gamma$ be an extensive-form game with SPNE $\boldsymbol{\sigma}$. Suppose there is a strategy $\sigma_i'$ for player $i$ such that  $\boldsymbol{\sigma}' = (\boldsymbol{\sigma}_{-i}, \sigma_i')$ is an SPNE and $u_i(\boldsymbol{\sigma}) = u_i(\boldsymbol{\sigma}')$ and there is player $j$ with $u_j(\boldsymbol{\sigma}) \neq u_j(\boldsymbol{\sigma}')$.
    Then $\Gamma$ is \emph{bossy at equilibrium}.
\end{definition}

\begin{theorem}\label{lemma:bossy}
    Online round-robin can be bossy at equilibrium. This is true even for instances with three agents with lexicographic preferences.
\end{theorem}
\begin{proof}
    Consider the two equilibrium allocations in the instance given by \Cref{tab:bossy}.
    As shown in the analysis of \Cref{lemma:multiplicity-of-equilibria}, agent $A$ is indifferent among the equilibrium allocations, while other players are not.
    By picking item $u$ or $w$ in her first turn, agent $A$ determines the equilibrium allocation.
\end{proof}

This violation is structurally surprising and normatively meaningful.
In our example, player $C$ would be incentivised to offer a large side-payment to player $A$ to select his preferred equilibrium.

\section{Equilibrium Fairness}
Equilibrium allocations hold surprising fairness guarantees in \emph{offline} round-robin as shown by \citet{amanatidis2024allocating}.
This motivates the study of equilibrium fairness at SPNEs of \emph{online} round-robin.
As is standard, we consider a relaxation of envy-freeness---envy-freeness up to one good (EF1)---which always exists and for additive agents is attainable by \emph{truthful} round-robin.

\begin{definition}
    An allocation $S_1, \ldots, S_n$ is \emph{envy-free up to one good} (EF1) if for all agents $i,j \in N$, we either have $S_j = \varnothing$ or there exists an item $g \in S_j$ such that $u_i(S_i) \geq u_i(S_j - g)$.
\end{definition}

\begin{restatable}{lemma}{lemmaStrategicAllocation}
    For additive utilities, in any SPNE of round-robin the first player receives utility of at least $\frac{u_1(M)}{n}$.
    \label{lemma:at-least-1-over-n}
\end{restatable}
\begin{proof}
    Consider a round robin instance with additive utilities, where $N$ is the set of players and $M$ is the set of items.
    The first player can receive a payoff that is at least $\frac{u_1(M)}{n}$ by playing truthfully.
    In each round, the first player picks its favourite item from the available items.
    Hence, any item that other agents pick in that round is less preferred than the item agent $1$ picked.
    So at the end of the division, for all agents $i \in N$, $u_1(S_1) \geq u_1(S_i)$.
    As $u_1$ is additive and $S_i$ disjoint, we have that $u_1(S_1) \geq \frac{u_1(M)}{n}$.
    By the definition of an SPNE, agent $1$ will act strategically to improve its utility, so the payoff of the truthful strategy acts as a lower bound for the utility of agent $1$ at an SPNE.
    Hence, at any SPNE, agent $1$ receives payoff that is at least $\frac{u_1(M)}{n}$.
\end{proof}

\begin{restatable}{theorem}{theoremTwoPlayerEF}
    For two players with additive utilities, any allocation that is the result of an SPNE is EF1.
\end{restatable}

\begin{proof}
    For two players, by \Cref{lemma:at-least-1-over-n} we have that at an SPNE $u_1(S_1) \geq \frac{u_1(M)}{2}$. As $u_1$ is additive, this in particular means that $u_1(S_1) \geq u_1(S_2)$, so agent $1$ is envy-free.
    Suppose that at an SPNE agent 1 picks item $a_1$. Then, we can apply \Cref{lemma:at-least-1-over-n} to the instance with items $M \setminus \{a_1\}$ where agent $2$ plays first.
    Hence, $u_2(S_2) \geq \frac{u_2(M \setminus \{a_1\})}{2}$ which implies that $u_2(S_2) \geq u_2(S_1 \setminus \{a_1\})$ and so agent 2 is envy free towards agent $1$ up to one good.
\end{proof}

Unfortunately, this result does not extend to instances with three agents.

\begin{theorem}
    There is a three-agent additive instance whose unique SPNE allocation is not EF1.
\end{theorem}
\begin{proof}
    \begin{table}
        \centering
        \begin{tabular}{llllllllll}
            \toprule
             &$a_1$ & $a_2$ & $a_3$ & $a_4$ & $a_5$ & $a_6$ & $a_7$ & $a_8$ &$ a_9$ \\ \midrule
             $A$ & 19 & 20 & 2 & 3 & 0 & 4 & 1 & 1 & 0 \\ \midrule 
             $B$ &  2  & 0 & 5 & 0 & 3 & 6 & 4 & 0 & 1  \\  \midrule
             $C$ & 13  & 14& 9 & 10& 5 & 0 & 11 & 14 & 0 \\ \bottomrule
        \end{tabular}
        \caption{The unique SPNE allocation of this instance violates EF1.}
        \label{tab:ef1-counterexample}
    \end{table}

    We analyse the instance given in \Cref{tab:ef1-counterexample} using computer-aided exhaustive search, the code for which appears in the ancillary files. The instance features a unique SPNE allocation: $\{a_8, a_1, a_2\}, \{a_9, a_3, a_6\}, \{a_4, a_5, a_7\}$ where the utility profile for agents is $(40, 12, 26)$.
    Notice that agent $C$ EF1-envies $A$: $u_C(\{a_1, a_2\}) = 27$, $u_C(\{a_1, a_8\}) = 27$ and $u_C(\{a_2, a_8\})= 28$, this is strictly greater than the value of agent $C$ for his own bundle.
    We also note that if we perturb the utilities by adding $2^{-k}$ to each agent's utility of item $a_k$, the unique equilibrium is unchanged. The SPNE allocation still violates EF1. 
    So, even if agents have strict preferences over bundles, the counterexample works.
\end{proof}

This result separates the fairness properties of equilibria of online and offline round-robin.
A counterexample of similar effect was also obtained in independent work by \citet{amanatidispersonalcommunication}.
On the positive side, we show that additive valuations satisfy a weaker notion of fairness: PROP1.

\begin{definition}[\citet{conitzer2017fair}]
    An allocation $S_1, \ldots, S_n$ is \emph{proportional up to one good} (PROP1) if for all agents $i \in N$, $u_i(S_i) \geq \frac{u_i(M)}{n}$ or there exists an item $g \in M \setminus S_i$ such that $u_i(S_i + g) \geq \frac{u_i(M)}{n}$.
\end{definition}

\begin{restatable}{theorem}{theoremProp}
    For additive instances, all SPNE allocations are PROP1.
\end{restatable}

\begin{proof}
    By \Cref{lemma:at-least-1-over-n}, player 1's utility is at least $\frac{u_1(M)}{n}$.
    Now, let us argue for player $i$.
    Suppose that at an SPNE, items $a_1, \ldots, a_{i-1}$ were picked before player $i$ plays for the first time.
    Then, this is equivalent to a new instance of round-robin where agent $i$ plays first and the set of items is $M \setminus \{a_1, \ldots, a_{i-1}\}$.
    So, by \Cref{lemma:at-least-1-over-n}, agent $i$ can guarantee a payoff of at least $u_i(M \setminus \{a_1, \ldots, a_{i-1}\})/n$.
    Suppose agent $i$ values $a_j$ the maximum amongst items $a_1, \ldots, a_{i-1}$.
    Then, $\frac{u_i(M \setminus\{a_1, \ldots, a_{i-1}\})}{n} + u_i(a_j) \geq \frac{u_i(M \setminus\{a_1, \ldots, a_{i-1}\}) + nu_i(a_j)}{n} \geq \frac{u_i(M)}{n}$.
    Hence, any allocation that arises from an SPNE satisfies PROP1.
\end{proof}

\section{Hardness for Submodular Utility Functions}

Now that we have shown structural properties of SPNEs, we focus on the computational complexity of computing them.
We define---and later show hardness for---the following promise problem.

\begin{definition}
    We define the promise decision problem \PromiseSPNE{k}, for $k$ a natural number greater than $1$.
    The input to the problem is the number of items $m$ and $k$ Boolean circuits $C_1, \ldots, C_k$ that represent the utilities of each agent.
    As part of the input there is also an agent $i$ and a threshold $v$. There is a promise that either for all SPNE $\boldsymbol{\sigma}$ of the instance, $u_i(\boldsymbol{\sigma}) \geq v$ or for all $\boldsymbol{\sigma}$, $u_i(\boldsymbol{\sigma}) < v$. The decision problem is to determine which of the two cases holds.
\end{definition}

Notice that this problem reduces to the search problem of finding an SPNE.
We can solve \PromiseSPNE{k} by searching for an SPNE and evaluating agent $i$'s utility at the SPNE allocation.
Also, \PromiseSPNE{k} reduces to deciding if an instance has an SPNE where agent $i$ receives utility at least $v$ trivially: it is the same decision problem without the promise.
So, showing hardness for \PromiseSPNE{k} shows hardness for both of these natural computational problems.

\begin{theorem}[\PSPACE-hardness for Submodular Utilities]
     \PromiseSPNE{2} is \PSPACE-hard even if one of the agents has a submodular utility function and the other has an additive utility function.
\end{theorem}
    We reduce from the \PSPACE-complete problem of Quantified Boolean Formula (QBF).
    The input to the decision problem is a quantified Boolean formula $\varphi = \exists x_1 \forall x_2 \ldots \exists x_{n-1} \forall x_n \psi$ where $\psi$ is a quantifier-free Boolean formula in CNF and $n$ is even. 
    The decision problem is to determine if $\varphi$ is true.

    Given $\varphi$ we construct a round-robin instance with two agents as follows.
    Agent $A$ controls the odd variables and wishes to satisfy $\varphi$ and agent $B$ controls the even variables and wishes to falsify $\varphi$.
    We simulate this by adding multiple items for each variable. For a variable $x_i$ we construct items $0_i$ and $1_i$.
    We add some buffer goods $g_i$ for $1 \leq i \leq n$ and two final goods $\varphi$ and $\neg\varphi$.

    Agent $B$ is additive and has weak lexicographic preferences over the items. Agent $B$'s preferences are given by the following utilities: $u_B(0_i) = u_B(1_i) = 2^{2n + 3 - 2i}, u_B(g_i) = 2^{2n + 2 - 2i}, u_B(\varphi)=2$ and $u_B(\neg\varphi)=1$.
    Agent $A$ values all singletons except for $\phi$ identically to $B$: $u_A(0_i) = u_A(1_i) = 2^{2n + 3 - 2i}, u_A(g_i) = 2^{2n + 2 - 2i}, u_A(\neg\varphi)=1$.
    Agent $A$ is also additive over all bundles that do not include $\varphi$.
    
    The final element of the construction is $A$'s utility for bundles containing $\varphi$.
    Some bundles of goods $S$ induce a Boolean assignment as follows: for $\ell \in \{0, 1\}$, if $\ell_i$ is in $S$, then $x_i \leftarrow \ell$.
    If this assignment is complete, no index appears twice and the assignment satisfies $\psi$, we say that $S$ satisfies $\psi$ and write $S \models \psi$.
    For brevity, let $n^* = n + \frac{n}{2}$ and $S^* = S - \neg \varphi$.
    Then, the marginal contribution of $\varphi$ to $S$ is:

    \begin{equation*}
        u_A(\varphi \mid S) =
        \begin{cases}
            \frac{8n + 1}{8n} & \text{if $\abs{S^*} < n^*$,}\\
            \frac{8n + 1}{8n} & \text{if $\abs{S^*} = n^*$ and $S^* \models \psi$,}\\
            \frac{8n-1}{8n} & \text{if $\abs{S^*} = n^*$ and $S^* \not\models \psi$,}\\
             \frac{8n + 1 - 4(\abs{S^*} -n^*)}{8n} &\text{otherwise.}
        \end{cases}
    \end{equation*}

    \begin{restatable}{lemma}{lemmaSubmodular}\label{lem:submodular}
        The utility function $u_A$ is monotone and submodular. The utility function $u_B$ is monotone and additive.
    \end{restatable}

\begin{proof}
    For $u_A$, by definition of submodularity, it suffices to show that for all $v \in M$, $S \subsetneqq T \subseteq M - v$ we have that $u_A(v \mid S) \geq u_A(v \mid T)$.

    Item $\neg \varphi$ has a marginal contribution of $1$ for all $S \subseteq M \setminus \{\neg \varphi\}$.
    From this fact, it follows that for all $v \in M \setminus \{\neg\varphi\}$ and $S \subseteq M$, $u_A(v \mid S) = u_A(v \mid S+\neg \varphi)$.
    Hence, when considering the marginal contributions of items $v$ other than $\neg \varphi$ we assume without loss of generality that $S \subsetneqq T \subseteq M \setminus \{v, \neg \varphi\}$.
    
    For item $\varphi$, and  $S \subsetneqq T \subseteq M \setminus \{\varphi, \neg \varphi\}$, $\abs{S} < \abs{T}$ implies $u_A(\varphi \mid S) \geq u_A(\varphi \mid T)$ by definition.
    
    Now, let us consider $v \in M \setminus\{ \varphi, \neg \varphi\}$ and $S \subsetneqq T \subseteq M \setminus \{v, \neg \varphi\}$.
    For all $S \subseteq M \setminus\{v, \neg\varphi\}$, if $\varphi \notin S$ then $u_A(v \mid S) = u_A(v)$.
    If on the other hand $\phi \in S$, we know that $u_A(v \mid S ) =  u_A(v) +u_A(\varphi \mid S+v -\varphi)- u_A(\varphi \mid S -\varphi)$, and we conclude that the possible value of $u_A(v \mid S)$ is one of $u_A(v)$, $u_A(v) - \frac{1}{4n}$ or $u_A(v) - \frac{1}{2n}$.

    If $u_A(v \mid S) = u_A(v)$, then $u_A(v \mid S) \geq u_A(v \mid T)$ as $u_A(v \mid S)$ attains the maximum possible value.
    If $u_A(v \mid S) = u_A(v) - \frac{1}{4n}$, then $\varphi \in S$ and $\abs{S + v - \varphi} \geq n^*$, hence $\abs{T -\varphi} \geq n^*$ and $u_A(v \mid T) \leq u_A(v) - \frac{1}{4n}$.
    Finally, if $u_A(v \mid S) = u_A(v) - \frac{1}{2n}$, then $\varphi \in S$ and $\abs{S + v - \varphi} > n^*$. Thus, $\abs{T} > n^*$ and $\varphi \in T$ and so $u_A(v \mid T) = u_A(v) - \frac{1}{2n}$.
    This concludes the proof that $u_A$ is submodular.
    Note, that it is also \emph{monotone}. The marginal contribution of $\neg \phi$ is always $1$, the marginal contribution of $v \in M \setminus \{\varphi, \neg \varphi\}$ is at least $2 - \frac{1}{2n}\geq 0$, and the marginal contribution of $\varphi$ is at least $\frac{1}{4}$.

    Utility $u_B$ is monotone and additive by definition.
\end{proof}

    \begin{restatable}{lemma}{lemmaSubmodularCorrectness}\label{lemma:submodular-example}
    
        At any SPNE, for even $0\leq k < n$ in round $i=3k+1$, items $0_i, 1_i, g_i$ will be available and agent $A$ will pick one of them.
        In round $3k+2$, agent $B$ will pick one of $0_i, 1_i$ and in round $3k+3$, agent $A$ will pick the remaining item from $0_i, 1_i, g_i$.
        For odd $k$, the above holds with $A$ and $B$ swapped.
        
    \end{restatable}

    \begin{restatable}{lemma}{lemmaSubmodularDecidingPhi}
        At any SPNE, agent $A$'s bundle $S_A$ has the following structure. For all $1 \leq i \leq n$, exactly one of $0_i$ and $ 1_i$ is in $S_A$. Also, for $1 \leq i \leq n$ and $i$ odd, $g_i \in S_A$.
        If in addition $\phi$ is true, item $\phi \in S_A$ and further the bundle $S_A$ induces a satisfying assignment of $\psi$.
        Otherwise $\neg\phi \in S_A$. Agent $A$'s bundle contains no other items.
    \end{restatable}
    
    Hence, \textsc{QBF} reduces to \PromiseSPNE{2}.
    For a given instance, either all SPNE allocations of $A$ satisfy $\psi$ or all do not.
    Let $v=u_A(\{y_1, \ldots, y_n, g_1, g_3, 
    \ldots, g_{n-1}, \phi\})$ where $y_i = 0_i$ or $y_i = 1_i$, and the $y_i$ induce a satisfying assignment of $\psi$.
    Hence, in all SPNEs, the utility of $A$ is at least $v$ or it is less than $v$ in all SPNEs. So, the promise is fulfilled.
    Deciding whether $A$ receives value at least $v$ decides whether the quantified Boolean formula $\varphi$ is true.

\section{Hardness for \OXS Utility Functions}

We now arrive at the most technically involved result of this paper.
We show \NP-hardness for the smallest superset of additive valuations studied by \citet{lehmann2001combinatorial} for a constant number of agents.
We describe the construction in detail, while presenting the proof of correctness in the appendix.

\begin{theorem}[\NP-hardness for \OXS utilities]
    \PromiseSPNE{11} is \NP-hard, even if the utility functions of agents are \OXS.
    \label{thm:np-hardness-oxs}
\end{theorem}

In this proof, agents have a particular type of \OXS utilities.
For all singletons $j$, agent $i$'s utility of $\{j\}$ will be a power of 2.
For each agent $i$ we partition the items of $M$ into $k$ components $S_1, \ldots, S_k$.
These components are unit-demand: for a bundle $S \subseteq S_j$, $u_i(S) = \max_{s \in S}(u_i(s))$.
Agent $i$'s utility across components is additive, formally $u_i(S) = \sum_{j = 1}^k \max_{s \in S_j \cap S} u_i(s)$.
All preferences of this form are \OXS utilities as per \Cref{def:OXS} with $f_j(S) = \max_{s \in S_j \cap S} u_i(s)$.
To elaborate on an example, suppose we write that agent $A$ has $a_1 > a_2 > a_3 > a_4$ preferences over goods and unit demand components of $\{a_1, a_3\}$ and $\{a_2, a_4\}$. This is shorthand for: $u_A(a_1) = 2^3, u_A(a_2) = 2^2, u_A(a_3) = 2^1,u_A(a_4) = 2^0$ over singletons.
But, $u_A(\{a_1, a_3\}) =2^3$ as both items are a subset of the unit demand component $\{a_1, a_3\}$. Agent  $A$ attains its maximum utility (up to ties) with bundle $\{a_1, a_2\}$ and receives utility $2^3+2^2$.

We construct gadgets for the preferences of agents.
We write these gadgets as functions of existing items. 
Throughout our definitions we will create fresh items.
These are new items and all agents outside the gadget have zero marginal utility towards them.

\begin{definition}
    For agents $A_1, A_2, A_3$ and goods $a, b, c$, we construct gadget \PickOneOf$(a > \{b, c\})$ as follows.
    We construct fresh items $u, v, x, y, z$. All agents $A_1, A_2, A_3$ have the same lexicographic preferences over singletons given by $a > u > v> b > c> x> y > z$.
    Agents have unit demand components in their utility functions.
    The non-singleton unit-demand components for $A_1$ are  $\{a, b, c, z\}$, for $A_2$, $\{u, b, c\}$ and $\{x, z\}$ and for $A_3$, $\{v, b, c\}$ and $\{y, z\}$.
\end{definition}

Under some conditions of non-interaction, if item $a$ is available, this gadget picks it, while items $b$ and $c$ remain available.
If instead $a$ is already picked, then $b$ and $c$ are picked. This is made formal below.

\begin{restatable}{lemma}{lemmaPickOneOfBraces}
    \label{lemma:pick-one-of-braces}
    Let $S$ be the set of items constructed for gadget $\PickOneOf(a > \{b, c\})$.
    Consider a state $q$ of round-robin where $A_1$ is next to play. Let the available items be $M$ and $S -a \subseteq M$.
    Suppose that for each item $s \in S$, for each $A_\ell$, $u_\ell(s) > u_\ell(M \setminus S).$
    Then, if $a \in M$, at every SPNE of $G_q$, agent $A_1$ will pick $a$ in state $q$.
    After all $A_\ell$ agents have taken two turns, items $b, c$ will be available and $z$ will be unavailable.
    If $a \notin M$, at every SPNE, after all $A_\ell$ agent have taken two turns, items $b,c$ will be unavailable and item $z$ will be available.
    After two turns, all $A_\ell$ will have marginal utility $0$ for all available items of $S$.
\end{restatable}

\begin{definition}
    For three agents $A_1, A_2, A_3$, existing items $a, b, c$ and fresh items $g_1, g_2, x, y, z$, gadget \PickOneOf$(a > b > c)$ captures the following preferences: $
        A_1: g_1 > g_2$,
        $A_2: g_1 > a > x > b > g_2 > y > z > c$,
        $A_3: a > x > b > y > z > c.$
    If an item is not mentioned in these preferences then it always has a marginal utility contribution of $0$.
    Agent $A_2$ has unit demand components of $\{a, b\}$ and $\{y, c\}$ and $A_3$ has unit demand components $\{x, b\}$ and $\{z, c\}$.
\end{definition}

Under similar conditions of non-interaction, this gadget ``picks'' the first item among $a, b$ and $c$.

\begin{restatable}{lemma}{lemmaPickOneOf}\label{lem:pickOneOf}
    Let $S$ the set of items constructed as part of the $\PickOneOf(a>b>c)$ gadget.
    Let $q$ a state where agent $A_1$ is next to play, the set of available items is $M$ and $S\setminus\{a, b, c\} \subseteq M$.
    Suppose further that for all $s \in S$ and all $A_\ell$, $u_\ell(s) > u_\ell(M\setminus S)$.
    Then, in any SPNE, in the next two turns of players $A_\ell$, the alphabetically first item from $a, b, c$ that is in $M$ will be picked. 
    The other items among  $b, c$ will not be picked.
    Also, if the picked item is $a$, it will be picked on the first turn of player $A_2$.
    All unpicked items from $S$ will have marginal value $0$ to all agents $A_\ell$.
\end{restatable}

For distinct items $a, b, c, d, e, f$, we can concatenate gadgets \PickOneOf$(a > b >c)$ and $\PickOneOf(d> e>f)$.
Let $S_1$ and $S_2$ the set of items in each gadget.
So, agents $A_1, A_2, A_3$ have preferences defined by \PickOneOf$(a > b >c)$ over items in $S_1$ and preferences defined by $\PickOneOf(d> e>f)$ over items in $S_2$. Agents $A_i$ have additively separable preferences between the items in different gadgets: $u_i(S) = u_i(S \cap S_1) + u_i(S \cap S_2)$.
However, by multiplying the value of every singleton in $S_1$ by an appropriate power of 2, we ensure that for all $s_1 \in S_1$, $u_i(s_1) > u_i(S_2)$.
As the sets of items in the gadgets are disjoint, the chaining defines valid preferences.
Because of the lexicographic preferences, we can apply \Cref{lemma:lexicographic-rr} freely and deduce that the gadgets will not interact with each other. We are now ready to prove \Cref{thm:np-hardness-oxs}.

    We reduce from \textsc{Satisfiability}.
    Consider an instance of \textsc{Satisfiability} $\varphi = \bigwedge_{i = 1}^m C_i$ with $m$ clauses and $n$ variables.
    \textsc{Satisfiability} remains \NP-hard even if each clause has size at most\footnote{Interestingly, if for all $i$, $\abs{C_i} = 3$, and each variable occurs at most 3 times, the formula is always satisfiable and so the problem is trivial \citep{tovey1984simplified}.} 3, i.e., $\abs{C_i} \leqslant 3$, and each variable occurs at most 3 times \citep{tovey1984simplified}.
    We make the following simplifying assumptions that maintain hardness: no variable appears twice in the same clause $C_i$ and every variable $x_i$ appears both negated and unnegated in the formula.
    
    We introduce one item for every occurrence of a literal in the formula.
    So, for a clause $C_i = x_a \lor \overline{x_b} \lor x_c$, we introduce items $x_a^i, \overline{x_b^i}, x_c^i$.
    We also introduce some buffer goods $g_1, \ldots, g_{3m}$.

    We construct agent $A$, who maximises his utility by satisfying the formula.
    Agent $A$ has lexicographic preferences over items.
    For $l_a^i$ the $a$-th literal appearing in clause $C_i$, the complete preference ordering of agent $A$ is:
$
        A\colon l_1^1 = l_2^1 = l_3^1 > g_1 > g_2 > g_3 > l_1^2 = l_2^2 = l_3^2 > g_4 > g_5 > g_6 > \cdots
 $.
    We also construct agent $B$ who will ``delay'' agent $A$.
    Agent $B$ has lexicographic preferences and a preference order $g_1 > g_2 > \cdots > g_{3m}$.

    Our construction will ensure that agent $A$ will not be able to pick two items that result in different assignments for the same variable, unless the formula is unsatisfiable.
    We achieve this with 9 additional agents, split in groups of three, they are named $P_1, P_2, P_3, Q_1, Q_2, Q_3, R_1, R_2, R_3$.
    Agents pick in order $A, P_1, P_2, P_3, Q_1, Q_2, Q_3, R_1, R_2, R_3, B$.
    We use agents $P_i$, $Q_i$ and $R_i$ to implement three chains of \PickOneOf{} gadgets.
    The details of $j$\textsuperscript{th} \PickOneOf{} gadget for agents $P_i$ are determined by the first literal of the $j$\textsuperscript{th} clause, $\ell^j_1$. For agents $Q_i$ and $R_i$ the details depend on the second and third literal of the clause, $\ell^j_2, \ell^j_3$ respectively.
    Some clauses may only have two literals. For those clauses, we instead chain lexicographic preferences  $a > b > c > d > e > f$ over fresh items to some agents $P_i, Q_i$ or $R_i$.
    The details of the \PickOneOf{} gadgets depend on whether this is the first, second, or third time the literal appears in the formula. Without loss of generality, we assume that the first instance of a variable appears unnegated.
    
    \paragraph{First appearance of variable $x_j$.}
    Suppose the first appearance of $x_j$ in formula $\phi$ is in clause $C_i$.
    If the literal $\overline{x_j}$ appears once in the formula in $C_a$ with $a > i$, we chain gadget \PickOneOf$(x_j^i > \{\overline{x_j^a}, g\})$ for a fresh good $g$.
    If instead $\overline{x_j}$ appears twice in the formula, say in clauses $C_a$ and $C_b$ with $a, b > i$, we chain \PickOneOf$(x_j^i > \{\overline{x_j^a}, \overline{x_j^b}\})$.
    We chain the \PickOneOf{} gadgets to agents $P, Q, R$, depending on if $x_j$ is respectively the first, second, or third literal of $C_i$.
    
    \paragraph{Second appearance of variable $x_j$.}
    Suppose the literal appears for a second time in clause $C_k$.
    If variable $x_j$ does not appear a third time, we simply append \PickOneOf$(x_j^k > a > b)$ for fresh items $a, b$.

    If a third literal appears in clause $C_l$, there are four cases to consider. By assumption, no variable will always appear unnegated. So, we can exclude ordering  $x_j, x_j, x_j$.
    If the literals appear in order $x_j, x_j, \overline{x_j}$, we use gadget \PickOneOf$(x_j^k > \overline{x_j^l} > a)$ for fresh item $a$.
    If the literals appear in order $x_j, \overline{x_j}, x_j$, we use gadget \PickOneOf$(\overline{x_j^k} > z > x_j^l)$, where $z$ is the item used in the first gadget that handled the first occurrence of $x_j$, i.e., \PickOneOf$(x_j^i > \{\overline{x_j^k}, g\})$.
    Item $z$ is available if and only if agent $P$ picked item $x_j^i$.
    If literals appear in order $x_j, \overline{x_j}, \overline{x_j}$ then we only need to use gadget \PickOneOf$(\overline{x_j^k} > a > b)$ for fresh $a$ and $b$.
    
    \paragraph{Third appearance of $x_j$.}
    For the third appearance of a literal, we simply need to ``gobble'' the item if it is available, so we can use \PickOneOf$(x_j^l > a > b)$ for fresh $a, b$.

    As we remarked earlier, it is not possible to chain \PickOneOf{} gadgets that share items, so each group $P, Q, R$ must consider each variable at most once.
    To do so, we pre-process the instance and assign each literal a colour $j\in \{1,2,3\}$
    The restriction is that every occurrence of the same variable is assigned a different $j$ and all instances of variables in a clause are assigned a different $j$. 
    Instances of variables will then appear in the preference lists of agents $P$ if $j=1$, $Q$ if $j=2$ and $R$ if $j=3$.
    Such an assignment always exists and can be computed in polynomial time by the below lemma.

\begin{restatable}{lemma}{graphColouring}
\label{lem:graph-colouring}
    Consider a graph $G = (V, E)$ with vertices $V \subseteq \{x_{i,j} \mid i \in [n], j \in [m]\}$ and for each $i \in [n]$, there are at most three occurrences of $x_{i,j'}$ in $V$ and for each $j \in [m]$ there are at most three occurrences of $x_{i', j}$.
    There is an edge between $x_{i, j}$ and $x_{i', j'}$ if and only if $i = i'$ or $j = j'$, but not if both equalities hold (i.e., no self-loops).
    Then, graph $G$ is 3-colourable and we can find such a 3-colouring in polynomial time.
\end{restatable}

\begin{proof}[Proof sketch]
We partition the set of edges into type A and type B. An edge between $x_{i, j}$ and $x_{i', j'}$ is of type A if $i=i'$, and of type B if $j=j'$. By construction of the graph, every vertex has at most two neighbours of type A and at most two neighbours of type B. Furthermore, we can make the following crucial observation: for every vertex, if it has two neighbours of type A (resp.\ type B), then those are also connected by an edge of type A (resp.\ B).

We can colour the graph in a greedy manner. In each round, we pick an uncoloured vertex $v_0$. If there is a colour that has not been assigned to any neighbour of $v_0$, we colour $v_0$ with that colour.
If that is not possible, we colour $v_0$ to one of the colours that appear once in its $4$ neighbours, without loss of generality colour $1$.
We label the neighbour of $v_0$ with colour $1$ as $v_1$. If we can recolour $v_1$ to a separate colour, then we are done.
If not, we recolour $v_1$ to a colour that is not $1$ that appears once in $v_1$'s neighbours.
We can repeat this simple process and by a parity argument on edges of type A and type B, we can show that this process never cycles and we reach a valid colouring.
\end{proof}
\begin{figure}
  \centering
  \begin{tikzpicture}[
  scale=1,
  line cap=round, line join=round, thick,
  vtx/.style={circle, fill=black, inner sep=1.6pt, outer sep=0pt}
]

\colorlet{myblue}{blue!70!black}
\colorlet{myred}{red!75!black}
\colorlet{mygreen}{green!60!black}

\node[vtx] (A) at (1,4.2) {};
\node[vtx] (B) at (2.4,3.5) {};

\node[vtx] (C) at (0.4,2.5) {};
\node[vtx] (D) at (-1.4,3.3) {};

\node[vtx] (F) at (-1.7,1.5) {};
\node[vtx] (G) at (-2,-0.5) {};

\node[vtx] (P) at (0,0) {};

\node[vtx] (I) at (2,1) {};
\node[vtx] (J) at (2,-1) {};

\path[myblue]
  (A) edge (B)
  (B) edge (C)
  (C) edge (A)
  (F) edge (G)
  (G) edge (P)
  (P) edge (F);

\path[myred]
  (D) edge (C)
  (C) edge (F)
  (F) edge (D)
  (P) edge (I)
  (I) edge (J)
  (J) edge (P);

\node[above left]  at (A) {$2$};
\node[below right] at (B) {$\cancel{1}\ \textcolor{mygreen}{3}$};
\node[below right]       at (C) {$\cancel{3}\ \textcolor{mygreen}{1}$};

\node[above left]        at (D) {$2$};
\node[left]        at (F) {\textcolor{mygreen}{$3$} $\cancel{1}$};

\node[below left]       at (G) {$2$};
\node[above]       at (P) {\textcolor{mygreen}{$1$}};

\node[above right] at (I) {$2$};
\node[below right] at (J) {$3$};

\node[below, text=gray, font=\normalsize] at (P) {$v_0$};
\node[right,  text=gray, font=\normalsize] at ($(F)+(0,-0.05)$) {$v_1$};
\node[above,  text=gray, font=\normalsize] at ($(C)+(-0.15,0)$) {$v_2$};
\node[above,       text=gray, font=\normalsize] at (B) {$v_3$};

\end{tikzpicture}
  \caption{An example of the recolouring procedure in the proof of \Cref{lem:graph-colouring}. Edges of type A are shown in blue, while edges of type B are shown in red. The original colouring of vertices is shown in black. Updates to the colouring are shown in green.}
  \label{fig:graph-colouring}
\end{figure}
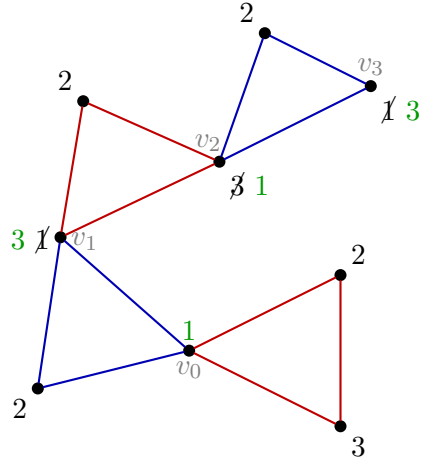
\begin{restatable}{lemma}{lemmaOxsCorrectness}\label{lem:OXS-correctness}
If $\varphi$ is satisfiable, at every SPNE, agent $A$ is allocated exactly one of the literal items for each clause and all goods of the form $g_{3k+2}$.
If $\varphi$ is not satisfiable, at an SPNE $A$ receives a bundle with strictly less utility.
\end{restatable}
From \Cref{lem:OXS-correctness}, we conclude that the promise of \PromiseSPNE{11} is satisfied. Additionally, deciding if the utility of agent $A$ is at least a threshold decides whether $\phi$ is satisfiable. This completes the reduction.
\section{Conclusion}
In this paper, we analyse the structural properties and computational complexity of SPNEs for subadditive agents.
Under strict conditions,  there may be exponentially many SPNE allocations and agents may violate non-bossiness when selecting them.
We show that equilibria need not be EF1 even for three additive agents, but recover the weaker PROP1 notion for additive instances.
On the computational front, we provide the first hardness results for a constant number of subadditive agents.
We show that one submodular agent is enough to make the two-agent case \PSPACE-hard.
Separately, by constructing significantly more involved gadgets, we are able to show \NP-hardness for 11 agents with \OXS utilities.
These are the first results for the computational complexity of SPNEs in round-robin for a constant number of agents since those of \citet{kohler1971class}.
\newpage
\bibliographystyle{plainnat}
\bibliography{arxiv}

\pagebreak
\appendix

\section{Omitted Proofs}
\label{sec:ommitted-proofs}

\subsection{Proof of Exponential Number of SPNE allocations}
\label{sec:exponential-spne-proof}
\begin{table}[h]
    \centering
    \begin{tabular}{lllllll}
    \toprule
        & $u$ & $v$ & $w$ & $x$ & $y$ & $z$ \\ \midrule
    $A$ & 32  & 16   & 8   & 4   & 2 & 1  \\ \midrule
    $B$ & 16   & 32  & 4   & 8   & 2 & 1  \\ \midrule
    $C$ & 8   & 4   & 2   & 32  & 16 & 1  \\ \bottomrule
    \end{tabular}
    \caption{A copy of \Cref{tab:bossy}. By copying this instance $k$ times, we can construct an instance with $\Omega(2^k)$ SPNE allocation.}
    \label{tab:bossy-2}
\end{table}

\exponentialAllocations*

    We construct an instance with agents $A, B, C$ and items $G^i = \{u^i, v^i, w^i, x^i, y^i, z^i\}$ for $i \in \{1,\ldots ,k\}$. We let the available set of items $M$ be $\bigcup_{i \in [k]}G^i$.
    For each agent, the utility for a good $g^i$ is $2^{6(i-1)}$ times the utility of good $g$ in the instance given by \Cref{tab:bossy-2}.
    For example, $u_A(v^6) = 2^{30} \times 16$.

    We claim that because of the lexicographic preferences, each copy of the instance given by \Cref{tab:bossy-2} doubles the number of equilibrium allocations. 
    We define these allocations formally.
    We use shorthands for particular bundles: $B_i^0 = \{v^i, y^i\}$, $B_i^1= \{v^i, x^i\}$, $C_i^0 = \{x^i, z^i\}$ and $C_i^1 = \{y^i, z^i\}$.
    In every allocation, agent $A$ always receives bundle $\{u^i, w^i \mid i\in [k]\}$.
    For each binary string $\vec{t} \in \{0, 1\}^n$ we define an allocation for agents $B, C$: agent $B$ receives $\bigcup_{i \in [k]} B_i^{t_i}$ and agent $C$ receives  $\bigcup_{i \in [k]} C_i^{t_i}$.
    This defines $2^k$ allocations. We claim that all such allocations are allocations that occur at some SPNE.

    The most natural direction would be to prove that all items in $G^k$ are picked before any other items. 
    Surprisingly, that is not the case. Hence, we need a more complex argument.

    \begin{lemma}
        \label{lemma:utility-of-A}
        At an SPNE agent $A$ receives utility of at least $u_A(\{u^k, w^k\})$.
    \end{lemma}
    \begin{proof}
        If agent $A$ picks $u^k$ as his first choice, by repeated application of \Cref{lemma:lexicographic-rr}, $B$ picks $v^k$ and $C$ picks $x^k$ and so agent $A$ can now pick item $w^k$ as required.
    \end{proof}

    \begin{lemma}
        \label{lemma:A-must-pick-Gk}
        Agent $A$ cannot pick an item not in $G^k$ in his first choice.
    \end{lemma}
    \begin{proof}
        Suppose agent $A$ picks in his first turn an item not in $G^k$.
        Then, agent $B$ can guarantee himself utility of $u_B(v^k + w^k)$ and agent $C$ can guarantee himself utility of $u_C(y^k + w^k)$ simply by picking their favourite available item.
        But, in an SPNE agent $A$ must receive utility at least $u_A(u^k + w^k)$ and agent $B$ must receive $v^k$.
        Hence, $\{u^k, w^k\}$ is in $A$'s bundle.
        So, $\{v^k, x^k\}$ must be in $B$'s bundle.
        But this contradicts that agent $C$ receives utility that is at least $u_C(y^k + w^k)$.
    \end{proof}

    \begin{lemma}
        Agent $A$ must pick item $u^k$ or $w^k$ in his first turn.
    \end{lemma}
    \begin{proof}
        By \Cref{lemma:A-must-pick-Gk}, we know that agent $A$ picks an item in $G^k$.
        If agent $A$ picks $v^k$ then agent $B$ must pick item $u^k$ in his turn by \Cref{lemma:lexicographic-rr}.
        So, $A$'s utility would be low contradicting \Cref{lemma:utility-of-A}.
        If agent $A$ picks one of $x^k, y^k, z^k$ then agent $A$ must also pick both $u^k$ and $v^k$ in later turns for \Cref{lemma:utility-of-A} to be fulfilled.
        But, that is not possible. Suppose for concreteness agent $A$ picked $x^k$, then Agent $B$ can guarantee himself utility of $u_B(v^k + y^k)$, similarly agent $C$ can guarantee himself utility of $u_C(u^k + z^k)$.
        However, at an SPNE agent $A$ must receive item $u^k$. So agent $C$ must receive utility of at least $u_C(y^k)$.
        So, agent $B$ cannot receive $y^k$ and so he must receive utility of at least $u_B(v^k + w^k)$. But $A$ must receive $w^k$ and $u^k$ leading to a contradiction.
        The cases where $A$ picks $y^k$ and $z^k$ in the first attempt are similar.
    \end{proof}

    \begin{lemma}
        If agent $A$ picks item $u^k$ in his first turn, then at any SPNE continuation, after each player has had two turns, agent $A$ has a partial allocation of $\{u^k, w^k\}$, agent $B$ a partial allocation $\{v^k, y^k\}$ and agent $C$ a partial allocation $\{x^k, z^k\}$. All items in $\bigcup_{i=1}^{k-1} G^i$ are unallocated.
    \end{lemma}

    \begin{proof}
        This is proved by repeatedly applying \Cref{lemma:lexicographic-rr} as in the proof of \Cref{lemma:multiplicity-of-equilibria}.
    \end{proof}

    \begin{lemma}
        If agent $A$ picks $w^k$ in his first turn, agent $B$ cannot pick an item that is not in $G^k$.
    \end{lemma}
    \begin{proof}
        By playing greedily, agent $B$ can guarantee himself a payoff of $u_B(v^k + y^k)$. If player $B$ does not pick an item in $G^k$, agent $C$ can guarantee himself a payoff of $u_C(x^k + v^k)$. But, agent $B$ must receive item $v^k$ and agent $A$ must receive item $u^k$ at an SPNE.
        Hence, agent $C$ must receive items $x^k$ and $y^k$.
        So, agent $B$ must receive payoff of at least $u_B(v^k +w^k)$.
        But this is impossible if agent $A$ must receive payoff of $u_A(u^k + w^k)$ by \Cref{lemma:utility-of-A}.
    \end{proof}

    \begin{lemma}
        If agent $A$ picks $w^k$ in his first turn, agent $B$ must pick item $x^k$.
    \end{lemma}
    \begin{proof}
        The maximum utility that agent $B$ can receive at an SPNE in this subgame is by receiving bundle $\{v^k, x^k\}$ from $G^k$.
        This is attainable by picking $x^k$.
        After agent $B$ picked $x^k$, agent $C$ can guarantee that he receives utility of at least $u_C(y^k + z^k)$.
        If agent $C$ picks either $u^k$ or $v^k$ in his next turn, agent $B$ will pick $y^k$ by repeated application of \Cref{lemma:lexicographic-rr}. Hence, item $v^k$ will be available at $B$'s next turn, thus guaranteeing he receives utility that is at least $u_B(v^k + x^k)$.

        This allocation is not possible if agent $B$ picks any other item from $G^k$.
        If $B$ picks $u^k$ in his first turn, then $A$ must pick $v^k$ in his next turn.
        If $B$ picks $v^k$ then at every SPNE, agent $C$ will receive item $x^k$, resulting in lower utility for agent $B$.
    \end{proof}

    \begin{lemma}
        If $A$ picks $w^k$ and $B$ picks $x^k$ then $C$ must pick one of $y^k, z^k$ in his next turn.
    \end{lemma}

    \begin{proof}
        As mentioned previously, picking $u^k$ or $v^k$ would result in another agent picking item $y^k$ thus reducing agent $C$'s utility below the $u_C(y^k + z^k)$ threshold.
        If agent $C$ picks an item outside of $G^k$, that means that agent $A$ can receive utility of at least $u_A(u^k + z^k)$ by picking greedily. In particular, that means that agent $C$ receives only one item from $G^k$ and that is below the $u_C(y^k + z^k)$ threshold.
        So, agent $C$ must pick $y^k$ or $z^k$.
    \end{proof}

    \begin{lemma}
        The picking order $w^k \rightarrow x^k\rightarrow y^k$ is the initial segment of an SPNE and must be continued by $u^k \rightarrow v^k \rightarrow z^k$.
    \end{lemma}

    \begin{proof}
        By repeatedly applying \Cref{lemma:lexicographic-rr}, we show that items are picked in order $u^k \rightarrow v^k \rightarrow z^k$. Hence, all items outside of $G^k$ are unallocated.
        This results in a partial allocation $\{u^k, w^k\}, \{v^k, x^k\}, \{y^k, z^k\}$.
    \end{proof}

    \begin{lemma}
        The picking order $w^k \rightarrow x^k\rightarrow z^k$ is the initial segment of an SPNE. In this SPNE items in $G^k$ are allocated as $\{u^k, w^k\}, \{v^k, x^k\}, \{y^k, z^k\}$. However, the continuation is not unique.
        Agent $A$ may either pick $u^k$ or an item from $G^{k-1}$ in his next turn.
    \end{lemma}
    \begin{proof}
        If items are picked in order $w^k \rightarrow x^k\rightarrow z^k$, then at every SPNE agent $A$ receives item $u^k$, agent $B$ receives item $v^k$ and agent $C$ receives item $y^k$ by \Cref{lemma:lexicographic-spne}.
        If at any point in the continuation any agent picks one of the items from $G^k$, the conditions of \Cref{lemma:lexicographic-rr} will trigger and the other two agents will also have to pick an item from $G^k$.
        Hence, we treat the items of $G^k$ as having their allocation predetermined, and having no strategic value.
        Now, as this is agent $A$'s turn, he can pick item $u^k$ which would trigger the cleanup operation.
        Alternatively, he can pick an item not in $G^k$.
        But, applying \Cref{lemma:A-must-pick-Gk} to the instance given by $\bigcup_{i=1}^{k-1} G^i$, we see that agent $A$ must pick an item in $G^{k-1}$ and in particular $u^{k-1}$ or $w^{k-1}$.
    \end{proof}

    These lemmas complete the proof.

\subsection{Hardness for Submodular Utility Functions}

\lemmaSubmodularCorrectness*

    \begin{proof}
    Consider the final allocations $S_A$ and $S_B$ of agents $A$ and $B$ at an SPNE.
    In its first two turns, agent $A$ can pick both of $0_1, 1_1$ or one of $0_1, 1_1$ and $g_1$. So, the utility of $A$ at an SPNE is at least $2^{2n+1} + 2^{2n}$.
    In its first turn, agent $B$ can always pick one of $0_1, 1_1$, thus at an SPNE agent $B$ has utility of at least $2^{2n+1}$.
    By the fact that agents have essentially lexicographic preferences, the above is only possible if $A$ picks two out of the three items $0_1, 1_1$ and $g_1$ and $B$ picks one of $0_1, 1_1$.
    $A$ must pick one of $0_1, 1_1, g_1$ in its first turn, and one in its second turn, otherwise $B$ will be able to pick two of these items in its first two turns.
    So, in the first turn $A$ must pick one of $0_1, 1_1, g_1$. If $A$ picks $\ell_1$ for $\ell \in \{0, 1\}$, $B$ must pick $(1-\ell)_1$ and $A$ must then pick $g_1$.
    If $A$ picks $g_1$ in its first turn, then $B$ picks one of $0_1,1_1$, and $A$ then picks the other Boolean value in its next turn. It is now $B$'s turn to pick.
    By applying the same reasoning, we see that $B$ must pick one of $0_2, 1_2, g_2$ in its next turn, then $A$ will pick one of $0_2, 1_2$ and finally $B$ will pick the remaining item.
    By applying this reasoning inductively, we see that at every SPNE, for all $i \in [n]$, agent $A$ and $B$ have exactly one of the $0_i, 1_i$ items. $A$ has all the odd numbered $g_i$ and $B$ has all the even numbered $g_i$.
    \end{proof}

\lemmaSubmodularDecidingPhi*

\begin{proof}
    From \Cref{lemma:submodular-example} we conclude that for odd $i$, agent $A$ can decide if $0_i$ or $1_i$ are in its bundle. It does so with full knowledge of $x_j$ for $j < i$.
    Agent $B$ similarly decides the value for $x_j$ and even $i$.
    Further, in the second to last turn, it will be agent $A$'s turn and the available items will be $\varphi$ and $\neg \varphi$. As this is its final choice, $A$ will simply pick the one that maximises its utility.
    At this point, agent $A$'s bundle induces a valid assignment on the $x_i$.
    If its induced assignment satisfies $\psi$, $A$ prefers item $\varphi$ to $\neg \varphi$.
    Otherwise, $A$ prefers item $\neg \varphi$ to $\varphi$.
    Agent $B$ will receive the item that agent $A$ does not pick. $B$ prefers item $\varphi$ to item $\neg \varphi$ and so has an incentive to falsify the formula.
    
    So, if $\varphi$ is satisfiable, $A$'s bundle at an SPNE will satisfy $\psi$. That is because in its turns, agent $A$ can force $\psi$ to be satisfied and doing so maximises its payoff.
    On the other hand, if $\varphi$ is not satisfiable, $A$'s bundle at an SPNE will not satisfy $\psi$, because agent $B$ can choose items that force $\psi$ to be false and thus maximise its payoff.
\end{proof}

\subsection{Hardness for \OXS{} Utility Functions}
\lemmaPickOneOfBraces*

\begin{proof}
    We can deduce the allocations by repeatedly applying \Cref{lemma:lexicographic-rr}.
    If $a$ is available, agent $A_1$ must pick it, otherwise agent $A_2$ would in its next turn.
    Similarly, agent $A_2$ must pick good $u$ in its turn and agent $A_3$ must pick good $v$.
    At this point, goods $b$ and $c$ have $0$ marginal utility to all agents, so they will not be picked.
    Afterwards, agent $A_1$ picks $x$, $A_2$ picks $y$, and agent $A_3$ picks $z$, resulting in final allocations of $S_1 = \{a, x\}, S_2 = \{u, y\}$ and $S_3 = \{v, z\}$.
    All agents have $0$ marginal utilities to the unallocated items $b, c$.

    If $a$ is unavailable, but $b$ and $c$ are, then by a similar analysis at the end of two rounds, agents receive the following bundles: $A_1 = \{u, c\}, A_2 = \{v, x\}$ and $A_3 = \{b, y\}$. The unallocated fresh item $z$ has no marginal utility to any agent.
\end{proof}

\lemmaPickOneOf*

\begin{proof}
    We proceed by repeatedly applying \Cref{lemma:lexicographic-rr}.
    If $a$ is available then, $S_1 = \{g_1, g_2\}$, $S_2 = \{a, y\}$, $S_3 = \{x, z\}$ and all agents have $0$ marginal utility towards $b, c$.
    In particular, if $a$ is available, the choices of the agents do not depend on whether $b$ and $c$ are available.
    If $a$ is unavailable but $b$ is, then $S_1 = \{g_1, g_2\}$, $S_2 = \{x, y\}, S_3 = \{b, z\}$ and both agents have $0$ marginal utility towards $c$.
    Finally, if $a$ and $b$ are unavailable, but $c$ is available, $S_1 = \{g_1, g_2\}, S_2 = \{x, z\}, S_3= \{y, c\}$.
\end{proof}

\graphColouring*

\begin{proof}
We partition the set of edges into two types. An edge between $x_{i, j}$ and $x_{i', j'}$ is of type A if $i=i'$, and of type B if $j=j'$. By extension, we also talk about neighbours of type A or B, depending on the type of edge connecting the vertex to a particular neighbour. By construction of the graph, every vertex has at most two neighbours of type A and at most two neighbours of type B. Furthermore, we can make the following crucial observation: for every vertex, if it has two neighbours of type A (resp.\ type B), then those are also connected by an edge of type A (resp.\ B).

We colour the graph in a greedy manner. In each round, we pick a vertex $v_0$ that has not yet been assigned a colour. If there exists a colour that has not been assigned to any neighbour of $v_0$, then we assign that colour to $v_0$ and proceed to the next round. Otherwise, if for each of the three colours there already exists some neighbour of $v_0$ with that colour, we need to modify the colouring along a particular path. We describe this next.

Without loss of generality, consider the case where vertex $v_0$ has two neighbours of type A, with colours 1 and 2, as well as one neighbour of type B with colour 3. Furthermore, if the vertex has a second neighbour of type B, and if that neighbour has been assigned a colour, we can safely assume that the colour is 2. Indeed, the colour cannot be 3 (since the two neighbours of type B are adjacent, by the crucial property), and if it is colour 1, then we can simply swap the roles of colours 1 and 2.

Now, let $v_1$ denote the neighbour of vertex $v_0$ that has colour 1, and change its colour to be 3. As a result, we can now assign colour 1 to vertex $v_0$, since none of its neighbours have that colour anymore. If there is no mistake in the colouring, then we can proceed to the next round. However, since we changed the colour of $v_1$, this might have introduced a mistake in the colouring, which we now need to fix. Let $v_2$ denote a neighbour of $v_1$ that has the same colour as $v_1$, namely colour 3. Note that $v_2$ is a neighbour of type B, and it is the only neighbour of $v_1$ that has colour 3. This follows from the crucial property above. Now, we change the colour of $v_2$ from 3 to 1. If $v_2$ does not have any neighbours with colour 1, we can proceed to the next round. Otherwise, let $v_3$ denote a neighbour of $v_2$ that has colour 1. By the crucial property, any neighbour of $v_2$ that has colour 1 must be of type A, since $v_1$ is a neighbour of type B and used to have colour 1 (and thus the other neighbour of type B cannot also have colour 1). Thus, $v_3$ is a neighbour of type A and it is the only neighbour of $v_2$ that has colour 1. The next step is to change the colour of $v_3$ from 1 to 3, and (if needed) consider its unique neighbour $v_4$ that has colour 3. This procedure must terminate with a valid partial colouring, unless it cycles, i.e., we end up changing the colour of a vertex a second time. \Cref{fig:graph-colouring} shows an illustration of this recolouring procedure. Next, we show that the procedure cannot cycle.

In order to prove that the procedure does not cycle, we note some properties of the path $v_0 v_1 v_2 v_3 \cdots$ on which we change the colouring. First of all, the edge $v_k v_{k+1}$ is of type A for even $k$, and of type B for odd $k$. Furthermore, when we reach $v_k$, we change its colour from 3 to 1 when $k$ is even, and from 1 to 3 when $k$ is odd. Finally, for every $k \geq 1$, any neighbours of $v_k$, excluding $v_{k-1}$ and $v_{k+1}$, must be uncoloured or have colour 2. Because of this last point, the only way the path could cycle (i.e., visit the same vertex more than once) is by reentering the path at $v_0$, i.e., $v_0 = v_k$ for some $k \geq 2$. In that case, $v_{k-1}$ would have to be a neighbour of $v_0$ that has colour 1 or 3, and is different from $v_1$. Thus, $v_{k-1}$ would have to be the neighbour of type B that has colour 3. Now, we have a contradiction. If $k-1$ is odd, then, as argued above, the colour of $v_{k-1}$ should be 1 (and we would change it to 3 when we reach it on the path). On the other hand, if $k-1$ is even, then, as argued above, the edge $v_{k-1} v_k$ should be of type A. So we either contradict the fact that $v_{k-1}$ has colour 3, or the fact that the edge $v_{k-1} v_k = v_{k-1} v_0$ is of type B. As a result, the path cannot cycle and terminates. At that point, the colouring is valid and we can proceed to the next round.

This greedy procedure yields a 3-colouring of the graph and it is easy to see that it runs in polynomial time, since in each round the path can visit each vertex at most once.
\end{proof}

\lemmaOxsCorrectness*

    \begin{proof}
        In the first turn, agent $A$ can pick one of the three literals of the clause $C_1$. If $A$ does not pick one of those three items, then they will be available to one of the agents in the \PickOneOf{} gadgets.
        These agents have lexicographic preferences, and these items are their favourite.
        So, by an argument similar to that of \Cref{lemma:lexicographic-rr}, agent $A$ must pick one of the items of $C_1$.
        Now, it is the turn of the $P$ agents. If agent $A$ did not pick the first literal of $C_1$, they will pick $\ell_1^1$ in their \emph{first turn} making it unavailable to $A$.
        Similarly, for agents $Q, R$ and items $\ell_2^1, \ell_3^1$.
        So, when it is the turn of agent $B$, items $\ell_1^1,\ell_2^1,\ell_3^1$ are unavailable, and so $A$'s favourite available item is $g_1$.
        Hence, we can apply \Cref{lemma:lexicographic-rr} to deduce that agent $B$ will pick item $g_1$.
        Now, it is the turn of agent $A$ again, and by \Cref{lemma:lexicographic-rr}, it must pick item $g_2$.
        Agents $P, Q, R$ take a second turn, and then agent $B$ takes item $g_3$.
        In these two turns, agents $P, Q, R$ have also chosen all items that represent the negations of the literal that agent $A$ has picked. Thus, whenever agent $A$ picks some literal, this ensures that it will not be able to pick any item representing the negation of that literal in later turns. In essence, this ensures that when a literal is used to satisfy a clause, its negation cannot be used to satisfy some other clause.
        
        Assuming the correctness of gadgets $P, Q,R$, we can proceed inductively and argue that if $\varphi$ is satisfiable, agent $A$ can make choices such that at least one of the items of clause $C_i$ is always available at turn $2i-1$.
        If agent $A$ acts in this way, then agents $P, Q, R$ will always pick any remaining available item from clause $C_i$ at turn $2i-1$.
        Hence, agent $A$ receives an allocation of exactly one item from each $C_i$ plus all items of the form $g_{3k+2}$.
        If instead $\varphi$ is \emph{unsatisfiable}, regardless of the choices of agent $A$, at some point there will be a clause $C_i$ such that agents $P, Q, R$ have already taken all items from it.
        As such, $A$ will receive strictly lower utility in the case where $\varphi$ is unsatisfiable.

        It remains to argue the correctness of the gadgets.
        \paragraph{First appearance of $x_j$.} If $x_j$ appears for the first time in clause $C_i$, gadget \PickOneOf$(x_j^i > \{\overline{x_j^a}, \overline{x_j^b}\})$ will pick $x_j^i$ if it is available on the first turn. In the second turn it will not pick any of the items $\overline{x_j^a}, \overline{x_j^b}$.
        If, instead, $x_j^i$ is unavailable, then items $\overline{x_j^a},\overline{x_j^b}$ must be available, so the agents of the gadget will pick them.
        Importantly, there is an item that can act as a ``flag'' for whether agent $A$ has picked item $x_j^i$. 
        By \Cref{lemma:pick-one-of-braces}, the fresh item $z$ is picked if and only if $x_j^i$ was not picked by $A$. This will be useful in the case of the second appearance of $x_j$.
        \paragraph{Second appearance of $x_j$.} If $x_j$ appears for a second time in clause $C_k$ and does not appear a third time, we simply need to make sure that item $x_j^k$ is not unavailable after the first turn of the agents of the gadget.
        This is achieved by \PickOneOf$(x_j^k > a > b)$ for fresh items $a, b$.
        Now, suppose $x_j$ appears in a literal in clause $l$ for the third time.
        If literals appear in order $x_j^i, x_j^k, \overline{x_j^l}$, then gadget \PickOneOf$(x_j^k > \overline{x_j^l} > a)$ behaves as expected.
        If $A$ does not pick item $x_j^k$ then the gadget will pick it in the first turn.
        If agent $A$ picks item $x_j^k$ then the gadget will pick its logical negation $\overline{x_j^l}$ after two turns if it is available.
        If instead item $\overline{x_j^l}$ is unavailable, because agent $A$ has picked the first instance $x_j^i$, then the gadget picks the fresh item $a$ which is always available.
        In the second case, where literals appear in order $x_j^i, \overline{x_j^k}, x_j^l$, we need to use gadget \PickOneOf$(\overline{x_j^k} > z > x_j^l)$, where $z$ is the item used in gadget \PickOneOf$(x_j^i > \{\overline{x_j^k}, a\})$.
        If item $\overline{x_j^k}$ is unavailable because agent $A$ picked item $x_j^i$, then item $z$ is also available. So the gadget picks item $z$ without affecting $x_j^l$.
        If instead, agent $A$ picks item $\overline{x_j^k}$ at turn $2k+1$, item $z$ is unavailable, and so the gadget picks $x_j^l$ which is the negation of the picked item.
        Finally, if the order of the literals is $x_j^i, \overline{x_j^k}, \overline{x_j^l}$, then there are no logical consequences to subsequent literals when agent $A$ picks $\overline{x_j^k}$ and the only function of the gadget is to ensure that the item is immediately unavailable.
        This is achieved by gadget \PickOneOf$(\overline{x_j^k}> a >b)$ for fresh $a, b$.
        \paragraph{Third appearance of $x_j$.} There are no logical consequences to subsequent literals when agent $A$ picks agent $x_j$. The only role of the gadget is then to ensure that the literal is not available after the first turn. This is achieved by the gadget $\PickOneOf(x_j^l>a>b)$ for fresh $a$ and $b$.
    \end{proof}

\end{document}